\documentclass[runningheads]{llncs}

\usepackage{graphicx}
\usepackage[leqno]{amsmath}
\usepackage{amssymb}

\usepackage{algorithm}
\usepackage{algorithmic}

\usepackage{xcolor}
\usepackage{url}

\newcommand{\R}{\mathbb R}

\newcommand{\stp}{\mathrm{STP}}
\newcommand{\fp}{\mathrm{STP}^{\downarrow}}

\begin{document}
\title{Exact approaches for the\\ Connected Vertex Cover problem}
%
%\titlerunning{Abbreviated paper title}
% If the paper title is too long for the running head, you can set
% an abbreviated paper title here
%
\author{Manuel Aprile\orcidID{0000-0002-6805-6903} }
\authorrunning{M. Aprile}
% First names are abbreviated in the running head.
% If there are more than two authors, 'et al.' is used.
%

\institute{Università degli studi di Padova\\
  Mathematics Department \\
  Via Trieste 63 Padova 35121, Italy\\
 \email{manuel.aprile@unipd.it}}
\maketitle              % typeset the header of the contribution

\begin{abstract}
Given a graph $G$, the Connected Vertex Cover problem (CVC) asks to find a minimum cardinality vertex cover of $G$ that induces a connected subgraph. In this paper we describe some approaches to solve the CVC problem exactly. First, we give compact mixed-integer extended formulations for CVC: these are the first formulations proposed for this problem, and can be easily adapted to variations of the problem such as Tree Cover. Second, we describe a simple branch and bound algorithm for the CVC problem. Finally, we implement our algorithm and compare its performance against our best formulation: contrary to what usually happens for the classical Vertex Cover problem, our formulation outperforms the branch and bound algorithm. 

\end{abstract}
\section{Introduction}
Given a graph $G=(V,E)$, a subset of vertices $C\subseteq V$ is a \emph{vertex cover} of $G$ if every edge of $G$ has at least one endpoint in $C$. The problem of finding a vertex cover of minimum cardinality in a graph is equivalent to finding a maximum stable set (or a maximum clique in the complement graph) and is one of the best studied problems in theoretical computer science. In this paper we study one of the most popular variants of the minimum Vertex Cover (VC) problem, where we aim at finding a minimum \emph{connected vertex cover} (CVC): i.e., we additionally require the subgraph $G[C]$ induced by $C$ to be connected. We call this the CVC problem.

 The CVC problem has applications in wireless network design, where one aims at placing relay stations on the network so that they cover all transmission links (the edges of the network) and are all connected to each other. 
 
Similarly to the VC problem, the CVC problem is NP-hard \cite{garey1977rectilinear} and admits a polynomial-time 2-approximation algorithm \cite{savage1982depth}. On the other hand, the CVC problem is NP-hard even if the input graph is restricted to be bipartite \cite{fernau2009vertex}: this is surprising as Vertex Cover is polynomially solvable for bipartite graphs, as, thanks to the famous König-Egeváry Theorem, it amounts to finding a maximum matching.

The CVC problem has received attention especially from the point of view of parameterized algorithms \cite{guo2007parameterized,molle2008enumerate} and approximation algorithms \cite{escoffier2010complexity,savage1982depth,cardinal2010connected}. 
An aspect that did not receive much attention is that of solving the CVC problem in practice: moreover, prior to this paper there were no mathematical programming formulations for the problem. Such formulations are usually easy to implement and are flexible to the addition of extra constraints to the problem, an advantage for real-world applications. 
Unlike for the CVC problem, there is a wealth of methods for solving the VC problem, the most effective being branch and bound algorithms %\cite{carraghan1990exact,li2010efficient}
(see \cite{wu2015review} for a survey), and there are many linear and non-linear formulations for VC and the related maximum clique and maximum stable set problems \cite{padberg1973facial,kleinberg1998lovasz,aprile2020extended}. 

A key feature of the CVC problem that we exploit in this paper is that its constraints can be modelled as linear constraints from two polytopes: the vertex cover polytope and the spanning tree polytope. Both are well-studied polytopes for which a large number of \emph{extended formulations} is known \cite{aprile2017extension,aprile2020extended,aprile2021smaller,fiorini2017smaller,Martin91,Wong80}: those are formulations where extra variables are used, other than the variables of the original polytope, in order to limit the number of inequalities.

In this paper we aim at partially filling the gap between VC and CVC by proposing mixed-integer extended formulations for the CVC problem.  Our main contribution is a mixed integer formulation for the CVC problem with a relatively small number of variables (linear in the number of edges of the input graph). 
The formulations we propose also lend themselves to modelling related problems as the Tree Cover problem \cite{arkin1993approximating} (see Section \ref{sec:con}). 
As an additional contribution, we also describe a simple branch and bound algorithm for CVC, by modifying a standard algorithm for the maximum stable set problem. Finally, we perform numerical experiments to compare the various approaches. In our experiments, the proposed mixed-integer formulation solves the problem much faster than the branch and bound algorithm. This is interesting since, for the general Vertex Cover problem, combinatorial algorithms usually outperform linear formulations.

The paper is organized as follows: this introduction terminates with Section \ref{sec:prelim}, which gives some basic terminology and notation; in Section \ref{sec:MIP} we give our formulations for CVC and prove their correctness; the branch and bound algorithm is described in Section \ref{sec:BB}; numerical experiments are given in Section \ref{sec:num}; finally, we conclude with some further research directions in Section \ref{sec:con}.

\subsection{Preliminaries}
\label{sec:prelim}
Throughout the paper we let $G=(V,E)$ be a connected graph. This is natural because, ignoring exceptions such as isolated vertices, only connected graphs admit connected vertex covers. A set $U\subseteq V$ is \emph{stable} if the subgraph $G[U]$ induced by $U$ does not contain any edge. Clearly, a subset $U\subseteq V$ is a vertex cover if and only if its complement $V\setminus U$ is stable. % hence solving the Vertex Cover problem is equivalent to solving the maximum stable set problem. Similarly
Hence, solving the CVC problem amounts to finding the maximum stable set $S$ such that the graph $G\setminus S$ obtained by removing $S$ is connected.
Finally, a subgraph of $G$ is a \emph{spanning tree} of $G$ if it is a tree and contains all vertices of $G$: we usually identify a spanning tree with a set of edges $F\subseteq E$.

For sets $U\subseteq A$, we denote by $\chi^U\in \{0,1\}^A$ the \emph{incidence} vector of $U$, which satisfies $\chi^U_v=1$ if and only if $v\in U$. We will use incidence vectors for subsets of vertices, edges, or arcs in directed graphs. For a vector $x\in \R^A$, we often write $x(U)$ to denote $\sum_{u\in U} x_u$.

\section{Mixed-Integer programming formulations}\label{sec:MIP}
A compact integer formulation of the Vertex Cover problem is well known: it suffices to use a variable $x_v$ for each node $v$ of our graph $G$, and ask that $x_u+x_v\geq 1$ for each edge $uv$ of $G$. On the other hand, it is not trivial to come up with a formulation for CVC, and we do not know any formulation that only uses node variables. The reason behind this difficulty is that imposing connectedness in an induced subgraph is a difficult constraint to model. Notice that a graph is connected if and only if it admits a spanning tree. Hence to model connectedness we resort to the spanning tree polytope of $G$, denoted by $\stp(G)$, defined as the convex hull of the incidence vectors of all the spanning trees in $G$. The basic idea that underlies all the formulations in this section is to add edge variables to the node variables, and to impose that such edge variables model a spanning tree in the subgraph induced by our vertex cover. We first propose the following formulation, based on the classical linear description of $\stp(G)$ given by Edmonds \cite{edmonds1971matroids}. 

\begin{align}
    \nonumber P_{\mathrm{stp}} = \Big\{x\in \{0,1\}^V \mid\:& \ \exists \ y \in [0,1]^E :\\
    &\quad x_u+x_v\geq 1 & \forall (u,v)\in E \label{constr:stp_cover}\\
    &y(E(U))\leq |U|-1 & \forall \emptyset \neq U\subseteq V \label{constr:stp_subset}\\
    &y(E)=x(V)-1\label{constr:stp_all}\\
    &y_{uv} \leq x_u, y_{uv} \leq x_v & \forall (u,v)\in E \label{constr:stp_link}
     \Big\}.
\end{align}

\begin{lemma}\label{lem:stp}
Let $G=(V,E)$ be a connected graph. Then $C\subseteq V$ is a CVC if and only if $(\chi^C, y)\in P_{\mathrm{stp}}$ for some $y\in \R^E$.
\end{lemma}

\begin{proof}
If $C$ is a CVC, then fix any spanning tree $F$ of $G[C]$. Then $\chi^C$ clearly satisfies Constraints \eqref{constr:stp_cover}; moreover, setting $y=\chi^F$ can be easily seen to satisfy Constraints \eqref{constr:stp_subset}, \eqref{constr:stp_all}, \eqref{constr:stp_link}.

On the other hand, assume that $(\chi^C, y)\in P_{\mathrm{stp}}$. Then $C\subseteq V$ is clearly a vertex cover. Moreover, \eqref{constr:stp_link} implies $y_e=0$ for each $e\in E\setminus E(C)$, hence the projection $y'$ of $y$ to variables $E(C)$ is in the spanning tree polytope of $G[C]$, due to constraints \eqref{constr:stp_subset},\eqref{constr:stp_all} (notice that $x(V)-1=|C|-1$). The spanning tree polytope of $G[C]$ is then non-empty, therefore $G[C]$ is connected. 
\end{proof}

The description above has an exponential number of constraints. There are well known extended formulations of size $O(n^3)$ for the spanning tree polytope of an $n$-vertex graph \cite{Wong80,Martin91}, and smaller extended formulations for special classes of graphs \cite{aprile2021smaller,fiorini2017smaller}. Therefore, we would like to turn any formulation for the spanning tree polytope into a formulation for CVC. This can be done by going through the forest polytope of $G$, $\fp(G)$, defined as the convex hull of incidence vectors of forests of $G$.
The same proof of Lemma \ref{lem:stp} shows that a correct formulation for CVC can be obtained by replacing Constraints \ref{constr:stp_subset} in $P_{\mathrm{stp}}$ with $y\in \fp(G)$. Finally, it is well-known that one can obtain a formulation of $\fp(G)$ from one of $\stp(G)$, since $\fp(G)=\{x\in [0,1]^E: \exists y\in \R^E: x\leq y, y\in \stp(G)\}$. While this approach does reduce the size of our CVC formulation from exponential to polynomial, it still yields too many extra variables to be practical. In the next section, we address this issue.

%Therefore, we would like to turn any formulation for the spanning tree polytope into a formulation for CVC.
%  In order to do that, we use the forest polytope $\fp(G)$, defined as the convex hull of incidence vectors of forests of $G$. It is well-known that one can obtain a formulation of $\fp(G)$ from one of $\stp(G)$, since $\fp(G)=\{x\in [0,1]^E: \exists y\in \R^E: x\leq y, y\in \stp(G)\}$. 

% \begin{align}
%     \nonumber P_{fp} = \Big\{x\in \{0,1\}^V \mid\:& \ \exists \ y \in [0,1]^E:\\
%     &x_u+x_v\geq 1 & \forall (u,v)\in E \label{constr:fp_cover}\\
%     &y \in \fp(G) \label{constr:y_in_fp}\\
%     &y(E)=x(V)-1\label{constr:fp_all}\\
%     &y_{uv} \leq x_u,\; y_{uv} \leq x_v & \forall (u,v)\in E \label{constr:fp_link}
%      \Big\}.
% \end{align}

% The following can be shown similarly as in the proof of Lemma \ref{lem:stp}.
% \begin{lemma}\label{lem:fp}

% %Let $G=(V,E)$ be a connected graph. Then $C\subseteq V$ is a CVC if and only if $(\chi^C, y)\in P_{fp}$ for some $y\in \R^E$.
% \end{lemma}

% Finally, it is well known that the integer formulation of vertex cover given by the edge constraints, which we use above, has a very weak linear relaxation, leading to slow solution times. However, one can replace constraints \eqref{constr:fp_cover} in $P_{fp}$ with any formulation for vertex cover, and still obtain a valid formulation for CVC. We will discuss this further in Section \ref{sec:num}. 

\subsection{A smaller mixed-integer formulation}
We now give a smaller formulation for the CVC problem, which makes use of a mixed-integer formulation for $\stp(G)$ with a small number of additional variables. We start by giving the formulation for $\stp(G)$, which builds on natural ideas that can be found, for instance, in \cite{miller1960integer}. Rather than spanning trees in undirected graphs, we focus on arborescences in directed graphs. Given our graph $G$, we simply bidirect each edge obtaining the directed graph $D(V,A)$. Now, fix a ``root'' vertex $r\in V$. Recall that an $r$-arborescence of $D$ is a subset of arcs $F\subseteq A$ such that, for every $v\in V\setminus \{r\}$, $F$ contains exactly one directed path from $r$ to $v$. Clearly, a description of the $r$-arborescences of $D$ gives a description of the spanning trees of $G$ by just ignoring the orientations (i.e. setting $y_{uv}=z_{uv}+z_{vu}$ for each edge $uv$). Moreover, since arborescences are rooted in $r$, we do not need arcs that point to $r$, and we simply delete them. %This is important for the correctness of the following formulation. 
Recall that $\delta^-(v)$ denotes the set of arcs of $A$ pointing to $v$.

\begin{align}
    \nonumber Q_r = \Big\{z\in \{0,1\}^A \mid\:& \ \exists \ d \in \R^V:\\
    &z(\delta^-(v))=1 & \forall v\in  V\setminus \{r\} \label{constr:arb1}\\
    &d_v\geq n\cdot(z_{uv}-1) + d_u +1 & \forall (u,v)\in A \label{constr:arb2}\\
        &d_r = 0 & \label{constr:arb3}\\
            &z(A)=|V|-1  \Big\}.\label{constr:arb4}
\end{align}

\begin{lemma}\label{lem:arb}
Let $D=(V,A)$ be a directed graph, and $r\in V$ such that $\delta^-(r)=\emptyset$. Then $F\subseteq A$ is an $r$-arborescence of $D$ if and only if $(\chi^F, d)\in Q_r$ for some $d\in \R^V$.
\end{lemma}
\begin{proof}
First, given an $r$-arborescence $F$, set $d_v$ to the length of the (unique) path from $r$ to $v$ in $F$, for each $v\in V$. It is easy to check that all constraints are satisfied by $(\chi^F,d)$.

On the other hand, let $(z,d)\in Q_r$, with $z=\chi^F$. We first show that $F$, after ignoring orientations, does not contain cycles: suppose by contradiction that $C\subseteq F$ is a cycle with vertices $v_1,\dots, v_k$, where for each $i=1,\dots,k$, $v_iv_{i+1}\in C$ or $v_{i+1}v_i\in C$ (where the sum is modulo $k$). For any $uv\in C$, we have that $d_v\geq d_u+1$ by \eqref{constr:arb2}: this implies that $C$ cannot be a directed cycle. In particular, if $v$ is the vertex of $C$ with $d_v$ minimum, then there are two arcs of $C$ pointing to $v$: but this is in contradiction with Constraint \eqref{constr:arb1}, if $v\neq r$, and with $\delta^-(r)=\emptyset$ otherwise.

Now, we have $|F|=|V|-1$ by \eqref{constr:arb4}. This, the absence of cycles, and Constraint \eqref{constr:arb1}, guarantees that $F$ is an $r$-arborescence of $D$.
\end{proof}

One could turn $Q_r$ into a formulation for the forest polytope of $G$ and obtain a formulation for the CVC problem, as described in the previous section. However, it is not clear how to do this without adding additional variables: the issue is the choice of the root $r$, which does not need to be connected to the other vertices in a forest. Instead, we are able to limit the number of variables by exploiting the fact that, for any edge $uv$ of $G$, at least one of $u, v$ has to be picked in our vertex cover. Hence, we choose a ``main'' root vertex $r$, and another root $r_1$, adjacent to $r$, that we can use as a root when $r$ is not in our vertex cover. We consider the following directed version $D(V,A)$ of our graph $G(V,E)$: fix $r,r_1\in V$ with $rr_1\in E$, turn every edge $vr\in E$ into a directed arc from $r$ to $v$, turn every edge $vr_1\in E$ with $v\neq r$ into a directed arc from $r_1$ to $v$, and bidirect each other edge. Notice that, in $D$, $\delta^-(r)=\emptyset$ and $\delta^-(r_1)=\{r\}$.
Now, consider the following formulation:

\begin{align}
    \nonumber P_{\mathrm{arb}}(r, r_1) = \Big\{x\in \{0,1\}^C: \mid\:& \ \exists \ z\in \{0,1\}^A, \ d \in \R^V:\\
    &x_u+x_v\geq 1 & \forall (u,v)\in A, \label{constr:arb_cover}\\
    &z(\delta^-(v))=x_v & \forall v\in  V\setminus \{r, r_1\} \label{constr:arb1'}\\
    &d_v\geq n\cdot(z_{uv}-1) + d_u +x_{v} & \forall (u,v)\in A \label{constr:arb2'}\\
        &d_{r} = 0 & \label{constr:arb3'}\\
       &z(A)=x(V)-1\label{constr:arb_all}\\
    &z_{uv} \leq x_u,\; z_{uv} \leq x_v & \forall (u,v)\in A \label{constr:arb_link}
     \Big\}.
\end{align}

\begin{theorem}
Let $G=(V,E)$ be a connected graph, let $r,r_1\in V$ with $(r,r_1)\in E$ and construct the directed graph $D(V,A)$ as described above. Then $C\subseteq V$ is a CVC if and only if $(\chi^C, z,d)\in P_{\mathrm{arb}}(r,r_1)$ for some $z,d$.
\end{theorem}
\begin{proof}
First, let $C\subseteq V$ be a CVC. We distinguish three cases.
\begin{enumerate}
    \item $r\in C, r_1\not\in C$. Let $F$ be any $r$-arborescence of $D[C]$, and set $x=\chi^C$, $z=\chi^F$, $d_v$ equal to the distance between $r$ and $v$ in $F$ for $v\in C$, and $d_v=0$ for $v\not\in C$. Notice that $0\leq d_v \leq n-1$ holds for all $v\in V$. Now, $(x,z,d)$ can be checked to satisfy all constraints of $P_{\mathrm{arb}}(r,r_1)$: we only discuss Constraints \eqref{constr:arb2'}.
    Let $(u,v)\in A$. If $(u,v)\not\in F$, the corresponding constraint is $d_v\geq -n + d_u +x_v$, which is trivially satisfied for any $u,v$ as $d_v$ is non-negative and the right-hand side is non-positive. Hence, suppose $(u,v)\in F$, hence $x_v = 1$. Then the constraint is $d_v\geq d_u+1$, which is satisfied at equality by our choice of $d$. 
    \item $r_1\in C, r\not\in C$. We proceed similarly as in the previous case, choosing an $r_1$-arborescence $F$ of $D[C]$ and setting  $z=\chi^F$, $d_v$ equal to the distance between $r_1$ and $v$ in $F$ for $v\in C$, and $d_v=0$ for $v\not\in C$. Then $(x,z,d)$ can be checked to satisfy all constraints exactly as before.
    \item $r,r_1\in C$. Let $F$ be an $r$-arborescence of $D[C]$ containing the arc $rr_1$ (notice that such an arborescence always exists). Set $z=\chi^F$, and set $d$ as in the first case. Again, one checks that all constraints are satisfied.
\end{enumerate}

Now, let $(\chi^C, z,d)\in P_{\mathrm{arb}}(r,r_1)$, with $z=\chi^F$. In order to show that $G[C]$ is connected, we just need to show that $F$ does not contain any cycle. We use the same argument as in the proof of Lemma \ref{lem:arb}, which we repeat for completeness. Assume that $F$ contains a cycle $C$. $C$ cannot be a directed cycle due to Constraints \eqref{constr:arb2'}, hence $C$ contains a vertex $v$ with two incoming arcs. Constraint \eqref{constr:arb1} implies that $v=r$ or $v= r_1$, but this contradicts the fact that $\delta^-(r)=\emptyset$, $\delta^-(r_1)=\{r\}$.
\end{proof}

\section{A Branch \& Bound algorithm}\label{sec:BB}
In this section we describe a naive branch \& bound algorithm to solve the CVC problem. For simplicity we follow the standard framework of branch \& bound algorithms for the maximum stable set problem, see for instance \cite{wu2015review}: instead of looking directly for a minimum vertex cover, we look for a stable set $S^*$ of maximum size. The only difference with the classical setting is that we impose that $S^*$ is feasible, where we call \emph{feasible} a stable set $S$ such that $G\setminus S$ is connected.

We now give an informal description of the algorithm, referring to Algorithm \ref{alg:BB} for the pseudocode.
To avoid recursion, a stack is used to store the nodes explored by the algorithm. Each node consists of a pair $(S,U)$, where $S$ is a feasible stable set and $U$ is a set of candidate nodes that can be added to $S$.
The idea is to explore the search space of all possible nodes while keeping a record of the best solution found so far, denoted by $S^*$: at each step, the current node $(S,U)$ of the stack is either branched on, or pruned if we realize that it cannot produce a stable set larger than $S^*$. The pruning step is based on greedy coloring, as in the classical algorithm for the maximum stable set problem, %\cite{babel1990branch}
 exploiting the fact that any proper coloring of the complement of a graph gives an upper bound on its maximum stable set: in particular, the maximum stable set that the node can produce has size at most $|S|+\alpha(G[U])\leq |S|+\chi(\bar{G}(U))$, and the latter term is estimated as the numbers of colors used in a greedy coloring (see Line \ref{line:while}). Branching is also performed as in the classical algorithm, but with a crucial difference: we select a vertex $v\in U$ and create nodes $(S,U\setminus \{v\})$ and $(S\cup\{v\}, U')$, where $U'\subseteq U\setminus \{v\}$ is obtained by removing from $U$ all the neighbors of $v$ \emph{and all the cut-vertices}\footnote{A vertex $v$ of a connected graph $G$ is a cut-vertex if its deletion disconnects $G$.} of $G\setminus (S\cup\{v\})$ (see Line \ref{line:branch}). This ensures that we only consider feasible stable sets.

 \begin{algorithm}[h]
	\begin{algorithmic}[1] 
        \REQUIRE{A connected graph $G=(V,E)$}

		\ENSURE{ A minimum-size CVC of $G$ }

		\STATE{$S^* \leftarrow \emptyset$}
		\STATE{$C \leftarrow $ cut-vertices of $G$}
		\STATE{$A\leftarrow [(\emptyset, V\setminus C)]$}
		
		\WHILE{$A$ non-empty}
		\STATE{$(S,U) \leftarrow$ pop($A$)}
		\WHILE{$U$ non-empty \AND $|S^*|<|S|+ $ greedy\_color($\bar{G}[U]$)}\label{line:while}
		\STATE{$v \leftarrow$ pop($U$)\label{line:pop}}
		\STATE{Append $(S,U)$ to $A$}
		\STATE{$S\leftarrow S\cup\{v\}$}
		\STATE{$C\leftarrow $ cut-vertices of $G\setminus S$}\label{line:cut}
		\STATE{$U\leftarrow (U\cap \bar{N}(v)) \setminus C$}\label{line:branch}
		\IF{$|S|>|S^*|$}
		\STATE{$S^*\leftarrow S$}
		\ENDIF
		\ENDWHILE
		\ENDWHILE
		\RETURN{$V\setminus S^*$}
	\end{algorithmic}
	\caption{Pseudocode of a basic branch \& bound algorithm for CVC. Following the classical framework for maximum stable set algorithms, the algorithm finds the largest stable set $S^*$ in $G$ such that $G\setminus S^*$ is connected, and then outputs the corresponding vertex cover.}\label{alg:BB}
\end{algorithm}

We now argue that our algorithm is correct: most importantly, we need to show that removing cut-vertices as described above is enough to find the largest feasible stable set. 

\begin{theorem}
Let $G=(V,E)$ be a connected graph. Then Algorithm \ref{alg:BB} on input $G$ outputs a minimum CVC of $G$.
\end{theorem}
\begin{proof}
Equivalently, we will show that the set $S^*$ output by the algorithm is the maximum feasible stable set of $G$. We say that a node $(S,U)$ \emph{contains} a feasible stable set $S'$ if $S\subseteq S'\subseteq U$. 

First, we claim that the starting node $(\emptyset, V\setminus C)$ contains all feasible stable sets, where $C$ are the cut-vertices of $G$. Indeed, if $u$ is a cut-vertex of $G$, and $S$ a feasible stable set, $S$ cannot contain $u$: if $u\in S$, we must have that $G\setminus \{u\}$ consists of two connected components $G_1$, $G_2$, and $S$ contains the vertices of $G_1$ without loss of generality. But since $G$ is connected, there is at least an edge between $u$ and a vertex of $G_1$, a contradiction.

Now, it suffices to show that, whenever we branch on a node $(S,U)$ obtaining two new nodes, any feasible stable set $S'$ contained in $(S,U)$ is contained in one of the new nodes. This implies that any feasible stable set is explored by the algorithm at some step, and concludes the proof.

The new nodes created are $(S,U\setminus \{v\})$ and $(S\cup\{v\}, U')$, where $U'$ is defined in Line \ref{line:branch}. Clearly, if $v\not\in S'$, then $S'$ is contained in node $(S,U\setminus \{v\})$ and we are done. On the other hand, if $v\in S'$, we only need to show that $S'\subseteq U'$. This follows since $S'$ cannot contain any neighbor of $v$, or any cut-vertex of $G\setminus (S\cup\{v\})$, where the latter is proved by using the same argument as for the starting node. 
\end{proof}

We conclude the section with some improvements to Algoritm \ref{alg:BB} that can be implemented to increase performance (see next Section for the implementation details).
\begin{itemize}
\item Computing a strong upper bound reduces the number of branch and bound nodes, at the price of longer running time for each node: for bipartite graphs, instead of resorting to a coloring bound we can directly compute the size of a maximum (usually unfeasible) stable set in the current subgraph, resulting in much better bounds and shorter total running time.
\item On the other hand, for general graphs we find that is better to spend less time on the upper bound computation: instead of recomputing a greedy coloring at each execution of Line \ref{line:while}, keeping the same coloring for several steps reduces the total running time.
    \item \emph{Russian Doll Search}: to slightly restrict the number of visited nodes, we order the vertices as $v_1,\dots, v_n$ by decreasing degree and call the algorithm $n$ times: at step $i$, we include node $i$ on our starting set $S$ and restrict the set $U$ to vertices $v_j$, with $j>i$, that are not neighbors of $v_i$.
    % \item \emph{Domain filter:} if adding a vertex $u\in U$ to the current stable set $S$ cannot lead to a  stable set larger than the current best $S^*$, because the degree of $u$ is too large, then $u$ can be safely deleted from $U$. In particular, after Line \ref{line:branch}, one can remove from $U$ all vertices $u$ such that $|U|-|N(u)|\leq |S^*|-|S|$, where $N(u)$ is the neighbourhood of $u$ in the subgraph of $G$ induced by $U$. 
\end{itemize}

\section{Numerical results}\label{sec:num}
We now compare the performance of our formulation $P_{\mathrm{arb}}$ and our branch and bound algorithm on a benchmark of random graphs. 
We remark that the CVC problem is most interesting in graphs where the solution of CVC is strictly larger than the minimum vertex cover (we call such graphs \emph{interesting}): if this is not the case one could just use the state of the art methods for finding the minimum vertex cover, and check that it induces a connected subgraph. 
This poses challenges to forming a benchmark of interesting graphs, as for instance the standard DIMACS benchmark \cite{johnson1996cliques} does not contain interesting graphs as far as we could check. 
%For instance, DIMACS instances \cite{johnson1996cliques} form a standard benchmark for maximum clique algorithms, and, once complemented, for maximum stable set and minimum vertex cover algorithms. However, all the instances that we could test were not interesting in the sense defined above.
 Hence we resorted to sparse, random graphs. In particular, half of our graphs are Erdős–Rényi random graphs with density equal to $0.05$; the others are bipartite random graphs, with density ranging from $0.1$ to $0.5$. We remark that bipartite graphs often seem to be interesting, which makes sense intuitively as each part of the bipartition forms a (possibly sub-optimal) vertex cover that is not connected: for instance, in the complete bipartite graph $K_{n,n}$, a minimum vertex cover has size $n$, while a minimum connected vertex cover has size $n+1$. Moreover, as mentioned in the introduction, bipartite graphs are one of the simplest graph classes for which the VC problem is polynomial and CVC is NP-hard, which makes them good candidates for studying the differences between the two problems. 

The graphs are produced with the functions fast\_gnp\_random\_graph() and bipartite.random\_graph() from the Networkx package \cite{hagberg2008exploring}, and the name of the graph indicates the random seed: for instance, $G_i$ is the random graph on 100 vertices with density 0.05 created by seed $i$. Some of the seeds are missing since we only consider connected graphs. 
The experiments are run on a processor Intel Core i5-4590 (4 cores) clocked at 3.3 GHz with 4 GB RAM. Algorithm \ref{alg:BB} is coded in Python, version 3.7, and Networkx functions articulation\_points() and greedy\_color() are used to perform lines \ref{line:cut} and \ref{line:while} respectively. %The order in which the algorithm processes the nodes of the graph is $v_1, \dots, v_n$, where $v_1$ is the node of minimum degree, and $v_i$ is the node of minimum degree in $G\setminus \{v_1,\dots,v_{i-1}\}$ for $i=2,\dots,n$. This is a standard order for such algorithms, and appears to lead to better running times than other orders we tried.
We refer to \cite{aprile2022github} for the code for Algorithm \ref{alg:BB} and for producing the formulation $P_{\mathrm{arb}}$.

As for the implementation of formulation $P_{\mathrm{arb}}$, it is also done in Python 3.7 and Gurobi 9.0.3 is used as MIP solver. Default parameters are used, and the results are averaged over three runs to account for the performance variability of the solver.  %As already discussed in Section \ref{sec:MIP}, the ``vertex cover constraints'' of $P_{\mathrm{arb}}$ make for a weak relaxation and several valid inequalities can be used to strengthen it. While this can significantly affect the solver performance for large graphs, it is not crucial for those examined in this paper. Hence, in our implementation we limit ourselves to greedily construct a set of maximal cliques that cover all the edges of the graph and add clique inequalities $x(C)\geq |C|-1$ (this is only done for non-bipartite graphs). 

\begin{table}
\centering
 \begin{tabular}{|c c c c c c c c|}%{|c{1.8cm} c{1.6cm} c{1.1cm} c{1.1cm} c{1.6cm} c{1.6cm} c{1.4cm} c{1.4cm}|} 
 \hline
 Name (seed) & $|V|$, $|E|$ & VC & CVC & B\&B t & B\&B n & $P_{\mathrm{arb}}$ t & $P_{\mathrm{arb}}$ n \\
 \hline
 $G_1$ & 100, 252 & 58 & 60 & 65.6 &  23138 & 0.2 & 1 \\
\hline
$G_2$ & 100, 247 & 55 & 56 & 6.4 &  435 & 0.15 & 1 \\
\hline
$G_3$ & 100, 232 & 56 & 57 & 12.7 & 1742 & 0.17 & 1 \\
\hline
$G_4$ & 100, 238 & 58 & 59 & 17.9 & 2296 & 0.43 & 191 \\
\hline
$G_{7}$ & 100, 257 & 56 & 59 & 21.1 & 2700 &  0.3 &  14\\
\hline
$G_9$ & 100, 254 & 58 & 60 & 100.3 & 21846 & 0.18 &  1\\
\hline
$G_{13}$ & 100, 260 & 58 & 59 & 56.2 & 18766 & 0.3 & 7\\
\hline
$G_{16}$ & 100, 263 & 56 & 58 & 18.1 & 3620 & 0.22 &  1\\
\hline
$G_{24}$ & 100, 234 & 58 & 58 & 11.2 & 1788 & 0.24  &  1\\
\hline
$G_{25}$ & 100, 264 & 61 & 61 & 28.6 & 4789 & 0.54 &  158\\
\hline
\end{tabular}
\caption{\label{tab:G}Results for random graphs of low density (0.05).}
\end{table}

\begin{table}
\centering
 \begin{tabular}{|c c c c c c c c|}%{|c{1.8cm} c{1.6cm} c{1.1cm} c{1.1cm} c{1.4cm} c{1.4cm} c c |} 
 \hline
 Name (seed) & $|V|$, $|E|$ & VC & CVC & B\&B t & B\&B n & $P_{\mathrm{arb}}$ t & $P_{\mathrm{arb}}$ n\\
 \hline
 $G_{0.1} (1)$ & 100, 255 & 49 & 54 & 11.18 &  6635 & 0.16 & 1 \\
\hline
$G_{0.1} (4)$ & 100, 242 & 50 & 57 & 863.4 &  818251 & 0.23 & 1 \\
\hline
$G_{0.2} (0)$ & 100, 483 & 50 & 57 & 1h+ & 2mln+ & 2.7 & 393 \\
\hline
$G_{0.2} (1)$ & 100, 497 & 50 & 56 & 1314.9 & 999252 & 2.3 & 338 \\
\hline 
$G_{0.3} (0)$ & 100, 753 & 50 & 55 & 1137.1 &  723409 & 4.2 & 88 \\
\hline  
$G_{0.3} (1)$ & 100, 753 & 50 & 55 & 1266.4 &  874949 & 4.3 & 166 \\
\hline
 $G_{0.4} (0)$ & 100, 1007 & 50 & 54 & 354.5 &  210209 & 3 & 1 \\
\hline
$G_{0.4} (1)$ & 100, 977 & 50 & 53 & 69.6 &  39614 & 2.1 & 1 \\
\hline
$G_{0.5} (0)$ & 100, 1254 & 50 & 53  & 73.5 &  38685 & 3.9 & 1 \\
\hline
$G_{0.5} (1)$ & 100, 1231 & 50 & 53 & 50.0 &  26071 & 5.4 & 1 \\
\hline
\end{tabular}
\caption{\label{tab:Bip}Results for random bipartite graphs. The density of each graph is written in its name, with the random seed in brackets.}
\end{table}

Table \ref{tab:G} indicates the results for random graphs, and Table \ref{tab:Bip} for bipartite graphs. Columns $VC$, $CVC$ indicate the sizes of the minimum vertex cover and connected vertex cover respectively. The columns B\&B t, B\&B n indicate the running time (in seconds) and the number of nodes of Algorithm \ref{alg:BB}, and similarly for $P_{\mathrm{arb}}$ t and $P_{\mathrm{arb}}$~n. 

It is evident from this comparison that solving the CVC problem with our formulation $P_{\mathrm{arb}}$ is much faster than with Algorithm \ref{alg:BB}, by a factor of one up to three order of magnitudes for some of the instances. Algorithm \ref{alg:BB} does not finish in the time limit (one hour) for one of the bipartite graphs of density 0.2. Clearly, this might be partially due to the naive implementation of Algorithm \ref{alg:BB}, which is not optimized for speed: for instance, in line \ref{line:cut} one does not have to recompute all cut vertices every time, but could restrict the computation to a single connected component of an appropriate subgraph of $G$. However, implementing this using the appropriate functions of Networkx actually further slows down the algorithm, as more information needs to be carried by each node. Hence, obtaining a faster version of the algorithm would require more advanced data structures and tools. But we believe this would not be enough to match the speed of $P_{\mathrm{arb}}$: a major limit of the algorithm is that the bound used in the pruning phase (line \ref{line:while}) is the same as for the classical vertex cover problem, i.e. does not take connectivity into account. Finding a better bound that is specific to the CVC problem is a non-trivial challenge, that we leave as an open problem.
On the other hand, since Gurobi solves $P_{\mathrm{arb}}$ using a very small number of branching nodes, it would seem that the bound of the linear relaxation of $P_{\mathrm{arb}}$ is reasonably tight. This suggests the idea of taking the best of both worlds and integrating a bound based on $P_{\mathrm{arb}}$ into a combinatorial branch and bound algorithm.

\section{Conclusion}\label{sec:con}
The CVC problem brings together two of the most natural concepts in graph theory: stable sets and vertex covers on one hand, connectedness and spanning trees on the other. This paper approaches the problem from a modeling perspective, giving exact mixed-integer formulations for solving the problem, and compares them with a simple branch and bound algorithm. We believe that further work needs to be done in both directions: while we focused on modeling the connectivity requirement, better formulations could be found by using tighter formulations of the vertex cover problem; on the other hand, finding a faster branch and bound algorithm is fascinating challenge, as it is unclear how to tailor the branching and pruning steps to the CVC problem.
We conclude by mentioning some extensions of CVC that could be of interest.

The Tree Cover problem \cite{arkin1993approximating} is closely related to the CVC problem: given a graph with non-negative weights on the edges and numbers $k, w$ one asks to find a connected vertex cover of size at most $k$ whose induced subgraph admits a spanning tree of weight at most $w$. It is easy to see that our formulations given in Section \ref{sec:MIP} can be adapted to model the Tree Cover problem, and exploring this further is an interesting research direction.

%: indeed, a feasible point in our formulations represents a connected vertex cover and a spanning tree (or arborescence) of the corresponding induced graph. In particular, given our formulation $P_{arb}(r_1,r_2)$ described in the previous section, one can model the Tree Cover problem by imposing that $x(V)\leq k$ and that $z(A)\leq w$. Similarly, one can solve the optimization version of the problem, where the weight of the spanning tree is minimized, by setting an appropriate objective function.
      
A natural generalization of the CVC problem considers hypergraphs instead of graphs \cite{escoffier2010complexity}. We remark that deciding whether a hypergraph contains a spanning tree is NP-hard \cite{andersen1995np}, hinting that the hypergraph version of CVC might be significantly harder than the graph version. However, we believe that our formulations can be extended to the hypergraph setting, and intend to investigate further in the future. %propose this as a second direction of research.

% : we recall that a hypergraph is a pair $H=(V,E)$ with vertex set $V$ and hyperedge set $E$, containing arbitrary subsets of vertices. Here, the notions of vertex cover and of connectedness can be defined in a natural way exactly as in graphs, however there are two different concepts of induced subgraphs that give rise to two versions of the problem, a ``strong'' and a ``weak'' one, both defined in \cite{escoffier2010complexity}.  However, the classical formulation of the spanning tree polytope can be adapted to hypergraphs by setting all variables to be integer \cite{warme1998spanning}. This can be seen to imply that our formulation $P_{\mathrm{stp}}$ can be adapted to the strong version of hypergraph CVC by using integer variables. We do not know whether our other formulations can be adapted in a similar way, and we leave this as an open question.

Finally, a different direction of research would be to generalize the connectivity constraint in the CVC problem to a matroid constraint, i.e. requiring that the edges of the subgraph induced by our vertex cover are full-rank sets of a given matroid. To the best of our knowledge, problems of this kind have not been studied before. Modelling such problems with mixed-integer formulations would be a promising line of inquiry, as there are several extended formulations for special matroid polytopes \cite{aprile2022regular,conforti2015subgraph,aprile2022extended}.
 \bibliographystyle{splncs04}
 \bibliography{mybib}

\begin{thebibliography}{10}
\providecommand{\url}[1]{\texttt{#1}}
\providecommand{\urlprefix}{URL }
\providecommand{\doi}[1]{https://doi.org/#1}

\bibitem{andersen1995np}
Andersen, L.D., Fleischner, H.: The np-completeness of finding a-trails in
  eulerian graphs and of finding spanning trees in hypergraphs. Discrete
  applied mathematics  \textbf{59}(3),  203--214 (1995)

\bibitem{aprile2022extended}
Aprile, M.: Extended formulations for matroid polytopes through randomized
  protocols. Operations Research Letters  \textbf{50}(2),  145--149 (2022)

\bibitem{aprile2022github}
Aprile, M.: Some code for solving the cvc problem (2022),
  \url{https://github.com/manuel-aprile/CVC}

\bibitem{aprile2020extended}
Aprile, M., Faenza, Y.: Extended formulations from communication protocols in
  output-efficient time. Mathematical Programming  \textbf{183}(1),  41--59
  (2020)

\bibitem{aprile2017extension}
Aprile, M., Faenza, Y., Fiorini, S., Huynh, T., Macchia, M.: Extension
  complexity of stable set polytopes of bipartite graphs. In: International
  Workshop on Graph-Theoretic Concepts in Computer Science. pp. 75--87.
  Springer (2017)

\bibitem{aprile2022regular}
Aprile, M., Fiorini, S.: Regular matroids have polynomial extension complexity.
  Mathematics of Operations Research  \textbf{47}(1),  540--559 (2022)

\bibitem{aprile2021smaller}
Aprile, M., Fiorini, S., Huynh, T., Joret, G., Wood, D.R.: Smaller extended
  formulations for spanning tree polytopes in minor-closed classes and beyond.
  Electronic Journal of Combinatorics  \textbf{28}(4),  P4.47 (2021)

\bibitem{arkin1993approximating}
Arkin, E.M., Halld{\'o}rsson, M.M., Hassin, R.: Approximating the tree and tour
  covers of a graph. Information Processing Letters  \textbf{47}(6),  275--282
  (1993)

\bibitem{cardinal2010connected}
Cardinal, J., Levy, E.: Connected vertex covers in dense graphs. Theoretical
  Computer Science  \textbf{411}(26-28),  2581--2590 (2010)

\bibitem{conforti2015subgraph}
Conforti, M., Kaibel, V., Walter, M., Weltge, S.: Subgraph polytopes and
  independence polytopes of count matroids. Operations research letters
  \textbf{43}(5),  457--460 (2015)

\bibitem{edmonds1971matroids}
Edmonds, J.: Matroids and the greedy algorithm. Mathematical programming
  \textbf{1}(1),  127--136 (1971)

\bibitem{escoffier2010complexity}
Escoffier, B., Gourv{\`e}s, L., Monnot, J.: Complexity and approximation
  results for the connected vertex cover problem in graphs and hypergraphs.
  Journal of Discrete Algorithms  \textbf{8}(1),  36--49 (2010)

\bibitem{fernau2009vertex}
Fernau, H., Manlove, D.F.: Vertex and edge covers with clustering properties:
  Complexity and algorithms. Journal of Discrete Algorithms  \textbf{7}(2),
  149--167 (2009)

\bibitem{fiorini2017smaller}
Fiorini, S., Huynh, T., Joret, G., Pashkovich, K.: Smaller extended
  formulations for the spanning tree polytope of bounded-genus graphs. Discrete
  \& Computational Geometry  \textbf{57}(3),  757--761 (2017)

\bibitem{garey1977rectilinear}
Garey, M.R., Johnson, D.S.: The rectilinear steiner tree problem is
  np-complete. SIAM Journal on Applied Mathematics  \textbf{32}(4),  826--834
  (1977)

\bibitem{guo2007parameterized}
Guo, J., Niedermeier, R., Wernicke, S.: Parameterized complexity of vertex
  cover variants. Theory of Computing Systems  \textbf{41}(3),  501--520 (2007)

\bibitem{hagberg2008exploring}
Hagberg, A., Swart, P., S~Chult, D.: Exploring network structure, dynamics, and
  function using networkx. Tech. rep., Los Alamos National Lab.(LANL), Los
  Alamos, NM (United States) (2008)

\bibitem{johnson1996cliques}
Johnson, D.S., Trick, M.A.: Cliques, coloring, and satisfiability: second
  DIMACS implementation challenge, October 11-13, 1993, vol.~26. American
  Mathematical Soc. (1996)

\bibitem{kleinberg1998lovasz}
Kleinberg, J., Goemans, M.X.: The lov{\'a}sz theta function and a semidefinite
  programming relaxation of vertex cover. SIAM Journal on Discrete Mathematics
  \textbf{11}(2),  196--204 (1998)

\bibitem{Martin91}
Martin, R.K.: Using separation algorithms to generate mixed integer model
  reformulations. Oper. Res. Lett.  \textbf{10}(3),  119--128 (1991)

\bibitem{miller1960integer}
Miller, C.E., Tucker, A.W., Zemlin, R.A.: Integer programming formulation of
  traveling salesman problems. Journal of the ACM (JACM)  \textbf{7}(4),
  326--329 (1960)

\bibitem{molle2008enumerate}
M{\"o}lle, D., Richter, S., Rossmanith, P.: Enumerate and expand: Improved
  algorithms for connected vertex cover and tree cover. Theory of Computing
  Systems  \textbf{43}(2),  234--253 (2008)

\bibitem{padberg1973facial}
Padberg, M.W.: On the facial structure of set packing polyhedra. Mathematical
  programming  \textbf{5}(1),  199--215 (1973)

\bibitem{savage1982depth}
Savage, C.: Depth-first search and the vertex cover problem. Information
  processing letters  \textbf{14}(5),  233--235 (1982)

\bibitem{Wong80}
Wong, R.: Integer programming formulations of the traveling salesman problem.
  In: Proc. 1980 IEEE International Conference on Circuits and Computers. pp.
  149--152 (1980)

\bibitem{wu2015review}
Wu, Q., Hao, J.K.: A review on algorithms for maximum clique problems. European
  Journal of Operational Research  \textbf{242}(3),  693--709 (2015)

\end{thebibliography}
\end{document}